\documentclass[11pt]{article}
\usepackage{amsthm, amsmath,amssymb,graphicx}

\newcommand{\comment}[1] {}

\newcommand{\eat}[1]{}

\advance \topmargin by -\headheight
\advance \topmargin by -\headsep
\evensidemargin \oddsidemargin
\newtheorem{theorem}{Theorem}

\newtheorem{definition}{Definition}
\newtheorem{lemma}{Lemma}

\newtheorem{claim}{Claim}

\newcommand{\D}{{\cal D}}
\newcommand{\E}{{\mathbf E}}
\renewcommand{\P}{{\cal P}}
\renewcommand{\S}{{\cal S}}

\title{Multi-Armed Bandit Problems with Delayed Feedback}

\date{} \author{Sudipto Guha\thanks{
Research supported in part by an NSF Award
CCF-0644119 and IIS-0713267.
Part of this work was done while the author
    was visiting Google Research. 
    Department of Computer and Information Sciences,
    University of Pennsylvania, Philadelphia, PA. Email: {\tt
      sudipto@cis.upenn.edu}} \and Kamesh Munagala\thanks{Supported by an Alfred P. Sloan Research  Fellowship,
and by NSF via CAREER award CCF-0745761 and grant CCF-1008065. Department of Computer Science,
    Duke University, Durham, NC 27708-0129. Email: {\tt kamesh@cs.duke.edu}}
 \and Martin P\'al\thanks{Google, 76 9th Ave, New York, NY. Email: {\tt mpal@google.com}} 
}
\begin{document}
\maketitle
\begin{abstract}
In this paper we initiate the study of optimization of bandit type
problems in scenarios where the feedback of a play is not immediately
known. This arises naturally in allocation problems which have been
studied extensively in the literature, albeit in the absence of delays
in the feedback. We study this problem in the Bayesian setting.
In presence of delays, no solution with provable guarantees
is known to exist with sub-exponential running time.

We show that bandit problems with delayed feedback that arise in
allocation settings can be forced to have significant structure that gives us the ability to
reason about this policy. We show a $O(1)$
approximation for a significantly general class of priors.
The structural insights we develop are of key interest and carry over to
the setting where the feedback of an action is available
instantaneously. In particular, we show a simple $2$-approximation for the finite horizon Bayesian bandit problem, improving and generalizing prior work. 
\end{abstract}

\thispagestyle{empty}
\newpage
\setcounter{page}{1}
\section{Introduction}
\label{sec:intro}
In this paper, we consider the problem of iterated allocation of
resources, when the effectiveness of a resource is uncertain, and is
learnt after some delay. Allocation of resources under uncertainty is
a central problem in a variety of disciplines, notably in learning
\cite{CL} and stochastic control \cite{book}. In these problems, we
are asked to make a series of allocation decisions, based on past
outcomes. Since the seminal contributions of Wald~\cite{Wald47} and
Robbins~\cite{Robbins}, a vast literature, including both optimal and
near optimal solutions, has been developed, see references in
\cite{Auer,Tsitsiklis,book,CL}. 
Of particular interest is the celebrated Multi-Armed Bandit (MAB) problem,
where an agent decides on allocating resources between competing
actions (arms) with uncertain rewards and can only take one action at
a time (play the arm). The agent collects the reward and the state of
the played arm is updated. 

However, an overwhelming majority of the literature focuses on scenarios that
assume {\em instantaneous} (or negligible in comparison to the
horizon) revelation of the outcome of each allocation.  In an early
work in mid 1960s, Anderson~\cite{Anderson} introduced the notion of
delayed observations.  Since then, though there have been additional
results~\cite{Suzuki,Choi}, a theoretical guarantee on adaptive
decision making under delayed observations has been elusive, and the
computational difficulty in obtaining such has been commented upon
in~\cite{sched,Armitage,Simon,Eick}. Recently, this issue of delays has been thrust to the fore due to the increasing application of iterative allocation problems in online advertising.

\subsection{Problem Statement and Motivating Examples}
We now concretely define the Bayesian multi-armed bandit problem with delayed feedback.  There is a bandit with $n$ independent arms. When arm $i$ is played, the reward is drawn $i.i.d.$ from distribution $D_i$, which is unknown. However, a prior $\D_i$ is specified over possible $D_i$. When arm $i$ is played, the feedback about its reward outcome is learned only after $\delta_i$ steps. The set of observed outcomes so far resolves the prior to a posterior distribution according to Bayes' rule. We also assume each arm has budget $B_i$ on the maximum reward it can accrue; further plays do not accrue reward. In this setting, the goal is to design a {\bf decision policy} for allocating the plays to the arms. A decision policy is a mapping from the current {\em state}, determined by the posterior distributions of each arm; the plays for which feedback is outstanding; and the remaining budgets of each arm to an action of which arm to play.  There is a horizon of $T$ steps, and our goal is to design a polynomial time algorithm that outputs a decision policy maximizing the expected revenue. 

\medskip
As concrete instantiations of this framework and a running example of the space of problems we address, consider the following three examples:

\medskip\noindent
{\bf 1. Online Determination of Web Content:} This is described by Agarwal etal in \cite{agarwal}. The goal is to present different websites/snippets to an user and induce the user to visit these pages. For each user (type) $j$ there is prior information about the propensity of the user to be interested in content $i$. This propensity is described by a distribution $D_{ij}$ drawn from a prior distribution $\D_{ij}$. The priors are constructed from historical estimation. As the pages are displayed, we resolve this prior; however, the characterization of the goodness a visit corresponds to multiple signals which is only known after a delay
$\delta_{ij}$. The delay can also arise from batched updates and systems issues. Note that there a number of ``budget constraints'' where each page can be displayed $T_j$ times.  For one user, this maps to the problem described above. The authors of~\cite{agarwal} present a number of different heuristics based on greedy policies.

\medskip\noindent
{\bf 2. Unmanned Aerial Vehicles:} Bandit  problems are used widely in stochastic control, where often actions are not immediate in their feedback. For example the difficulty in routing multiple UAVs among an uncertain terrain  is the remoteness of the controller \cite{leny}, which induces delays.

\medskip\noindent
{\bf 3. Budgeted Allocation Problem:} In this problem, there are $n$ bidders (advertisers competing for a collection of $m$ impressions (advertisement slots). Assume that the impressions arrive one at a time, and correspond to the "semi-online" framework where the number of impressions of type $j$ is $T_j$ and known beforehand.  Each bidder $i$ has a budget $B_i$ on the total amount she is willing to pay, and bids $b_{ij}$ for impression of type $j$.  When advertiser $i$ is allocated an impression of type $j$, she is charged her bid $b_{ij}$ only if an acquisition occurs. This is referred to as the cost-per-action (CVR) model. The search engine allocates the impressions based on its estimate of the click or conversion probability (CVR) $p_{ij}$ of this favorable event, in order to maximize its expected revenue (where if a favorable event occurs for pair $(i,j)$, and $i$'s budget is not exhausted, the contribution to revenue is $b_{ij}$). The classical budgeted allocation problem~\cite{MSVV,free,ChakrabartyGoel,DH09} assumes that (i) the $p_{ij}$ are known; and (ii) feedback about the click or conversions is instantaneous. Stochastic relaxations of the assumption (i) is already considered in \cite{PandeyOlston}, who also provide simulations of  greedy algorithms and consider budgets. But in case of conversions 
(and in case of more refined measurements of clicks as well) there is a natural delay in the feedback, which previous literature largely ignores.

\subsection{Our Results and Techniques}
We present a constant factor approximation to the Bayesian MAB problem with delayed feedback. Before presenting our results, it is instructive to study closely related work. The Bayesian MAB problem has been considered in the model where the feedback is instantaneous, see~\cite{moore,madani,GM09,GKN}, and constant factor approximations are known in this setting. In the absence of delays, all approximation algorithms (as well as the exact Gittins index algorithm in the discounted reward setting~\cite{Gittins,Tsitsiklis}) use the key observation that the posterior distribution of an arm does not change if not played. Using this observation,  it was shown in~\cite{GM09} that there is an approximately optimal policy that sequentially plays the arms. In presence of delays, the state of an arm now needs to capture the plays with outstanding feedback, and this not only causes the state to change even with no plays, but also requires exponential space to specify the state even for a single arm. Finally, any policy has to now interleave plays of the different arms, and hence, existing techniques in~\cite{GM09,GKN,Tsitsiklis} do not apply in any straightforward fashion.  In a sense, the challenge in the delayed setting is not just to balance explorations along with the exploitation, {\em but also to decide a {\bf schedule} for the possible actions both for a single arm (due to delays) as well as across the   arms}.

 At a high level, our algorithm starts with the standard LP relaxation to the problem~\cite{book,GM09}. This relaxation  provides a collection of {\em single-arm} policies whose execution is confined to one bandit arm; the LP enforces constraints on the total expected plays these policies can make. Such a relaxation has been widely used in the absence of delays, since it leads to efficient {\bf decomposable decision policies} (called index policies) that treats every arm independently of other arms.  Note that the joint state space over all the arms is exponential in size. As alluded to above, in the presence of delays, it is not even clear that such decomposable policies are efficiently computable. The most important question that we ask for delayed feedback therefore is: {\em Do there exist near optimum policies that are
decomposable and efficiently computable?} In this paper we develop new techniques to analyze single-arm policies and provide a poly-time algorithm that outputs a decomposable decision policy\footnote{These policies are very close to index policies; the policy of each arm is independent of other arms and hence efficiently computable; however, the policies of different arms cannot be  reduced to a single {\em priority} value.} which is an $O(1)$ approximation of the possibly entangled optimum policy. 

In the rest of the paper we focus on the MAB case; we show in Appendix~\ref{part3} how the results generalizes to multiple MABs or users (budgeted allocation case). This generalization is a consequence of the fact that the algorithm uses the decomposition property (which implies linearilty).

We summarize our main results and technical contributions below:
\begin{itemize}
\item In Section~\ref{part1}, we prove a structural result, {\em Truncation Theorem}. Theorem~\ref{statetheorem}  shows that for arbitrary  state spaces, even in the presence of budgets, delays and other constraints which preserves paths, a partially executed single-arm policy has reward  proportional to the original policy.  This uses a stopping time type  argument and uses the connection between state spaces corresponding to priors and their Martingale properties. This idea leads to shortening the horizon of single-arm policies, and is very general; it directly yields a $2$-approximation for the finite horizon MAB problem with instantaneous feedback. In the absence of budgets, the best known result for this problem was a $12$ approximation~\cite{GM09,GKN}; their approach is infeasible in the presence of budgets.  Our result also shows a factor $2$ approximation for the "irrevocable bandit" problem\footnote{The algorithm is not allowed to revisit an arm.}, improving the factor $8$ result of Farias and Madan \cite{farias}. We also show that factor $2$ is the best possible bound against LP relaxations over single arm policies. This result  has also been used in subsequent  results for MABs with nonlinear objective functions~\cite{GM11b}.

\item In Section~\ref{part2}, using the Truncation theorem, we show a $O(1)$ approximation for  the MAB problem with delayed feedback when $\delta_{i} =  o\left(T/\log T\right)$.\footnote{We can show that using  stochastic domination of Gittins-type indices, we can increase $\delta_i$ to $O\left(\frac{T}{\log \log T}\right)$ and still preserve the $O(1)$ approximation of decomposable policies. However, the proof is very technical and omitted.} Applying the Truncation theorem is not straightforward - a policy could make very few plays initially waiting for feedback after every play, and load all the plays at the very end when it is fairly certain the reward is large. To circumvent this, we develop an interesting compaction and simulation argument - we modify the policy to make more plays upfront, but withhold using the outcomes  until the original policy used them. The resulting structure also implies that our policies are as efficient to compute as standard index policies (without delays)~\cite{Gittins,book}. Finally, the combination of the single-arm policies into a final feasible policy requires a novel priority based scheme. 
\end{itemize}

The chief technical highlight of our work is a new way of accounting for the reward of a single-arm policy using the martingale property of priors that Bayes' rule entails. Traditionally, the Markov Decision Process (MDP) formulation resulted in the reward being accounted non-uniformly: given a posterior distribution, a play yields reward which is the expected value of the posterior. We use the fact that the plays draw $i.i.d.$ rewards from a fixed underlying distribution (whose current belief is encoded in the prior), and hence the expected reward is the same for each play conditioned on this unknown distribution. We use this accounting for both the truncation and compaction steps, leading to simplicity of policies, clarity of analysis, and improved approximation bounds. 

\subsection{Other Related Work}
\label{related}
There is an extensive literature on the MAB problem in the prior-free setting; see~\cite{LaiRobbins,Auer,CL}), and policies with additive {\em regret} guarantees are known. Regret is the difference between the expected reward of the policy and the reward of an omniscient policy which knows all the distributions. However, these results both require the reward rate to be large and large time horizon $T$ compared the number of arms. In the application scenarios mentioned above, it will typically be the case that  the number of arms is very large and comparable to the optimization horizon and the reward rates are low. This motivates the need for a purely multiplicative guarantee instead of additive guarantees. Moreover the analysis of these policies require the plays of the arm with maximum estimated reward to be continuous, which is not true in presence of delays (or budgets). If the delay of arm $i$ satisfies $\delta_{i} =o(n^{1/3}T^{2/3})$ (where $n$ is the number of arms) then standard explore-then-exploit (where we play an arm long enough to start receiving the outcomes) strategy gives sub-linear in $T$ regret. However, constants in the regret term depends on scaling of the rewards.

\renewcommand{\E}{\mathbf{E}}
\newcommand{\Pp}{{{\cal P}}}
\newcommand{\n}{{\overrightarrow{n}}}

\section{Preliminaries}
\label{prelim}
There is a bandit with $n$ independent arms. The arm $i$ underlying reward distribution $D_i$, which is a random variable drawn from a prior distribution $\D_i$. These priors are specified as input. The maximum possible reward that can be extracted from this arm is $B_i$, and though additional plays can be made, they do not accrue additional reward. If an arm is played, the feedback about the reward outcome is
available only after $\delta_i$ time steps. As observations are available the successive posteriors (which serve as priors for the next trial), are produced by the Bayes' Rule. There is a time horizon of $T$ plays. A {\bf decision policy} specifies which arm to play given the current state of each arm, which is captured by the remaining budget and time horizon, posterior distribution, and plays with outstanding feedback for that arm. Each decision policy has a unique expected reward value
that is obtained on executing it (with different execution trajectories differing due to different realizations of the underlying $\{D_i\}$).  {\bf Our goal} is to design a poly-time algorithm, which outputs a policy that approximately maximizes the expected value
derived over the horizon of $T$ plays.  

\newcommand{\p}{\mathbf{p}}
\medskip
\noindent{\bf Input: Priors and Posterior Spaces.} For each arm $i$, the input specifies the space $\S_i$ of possible posterior distributions, which defines a natural DAG, the root of which corresponds to the initial prior $\D_i$.  Every other state $u \in \S_i$ corresponds to a set of observations, and hence to the  posterior $X_{iu}$ obtained by applying Bayes rule to $\D_i$ with the those observations. Playing the arm in state $u$ yields a transition to state $v$ with probability $\p_{uv}$, provided $v$ can be obtained from $u$ in one additional observation; the probability $\p_{uv}$ is simply the probability of this observation conditioned on the posterior $X_{iu}$ at $u$. The expected posterior mean at a state $u \in \S_i$, denoted by $r_u = \E[X_{iu}]$ satisfies the {\bf martingale} property $r_u = \sum_v \p_{uv}r_v$. An example  is in Appendix~\ref{prelim2}.

We need the running time to be poly$(n,T, \sum_i |\S_i|)$, which is comparable to the running times (via per-arm dynamic programming) for computing the standard index policies~\cite{book,Gittins}.

\medskip
\noindent{\bf Budgets:} The budget $B_i$ for arm $i$ can be folded into the description of the posterior space: If the observations leading to the current posterior already violate the budget, the reward of this state is set to $0$. This transformation helps us ignore the budget  in subsequent sections.

\medskip
\noindent{\bf Single-arm Policies.} Given an execution of the global policy $\Pp$, define its {\bf projection} on arm $i$ to be the policy $\Pp_i$ defined by the actions induced on $\S_i$; we term this a {\em single-arm} policy. Note that the global policy may take an action in arm $i$ based on information regarding other arms (an entangled state) -- that side information is lost in the projection. 

Let $\{\Pp_s(i,T)\}$ describe all  single-arm policies of arm $i$ with feedback delay of $\delta_{i}$ steps and horizon $T$.  Each $\Pp_s(i,T)$ is a (randomized) mapping from the current state to one of the following actions: (i) make a play; (ii) wait some number of steps (less or equal to $T$), so that when the result of a previous play is known, the policy changes state; (iii) wait a few steps and make a play (without extra information); or (iv) quit.  The {\em state} of the system is captured by the current posterior $u \in \S_i$, the plays with outstanding feedback, and the remaining time horizon. Note that the state encodes plays with outstanding feedback, and this has size $2^{\delta_i}$, which is exponential in the input. Ameliorating this dependence  is an aspect that we address in our algorithm design.

\begin{definition}
Given a single-arm policy $\Pp_i$ let $R(\Pp_i)$ to be the   expected reward and $N(\Pp_i)$ as the expected number of plays, where the expectation is over the outcomes of the plays of $\Pp_i$. 
\end{definition}

\section{The Truncation Theorem}
\label{part1}
We now show that the time horizon of a single-arm policy on $T$ steps can be reduced to $\beta T$ for constant $\beta \le 1$ while sacrificing a constant factor in the reward. We note that though the statement seems simple, this theorem {\em only} applies to single-arm policies and is {\em not true} for the global policy executing over multiple arms.  The proof of this theorem uses the martingale structure of the rewards, and proceeds via a stopping time argument.

\begin{theorem}
\label{statetheorem}(The Truncation Theorem)
For any arbitrary single arm policy $\P$ which traces out a path over a horizon of $T$ steps in a state space $\S$, then the identical policy $\P'$ that makes at least a $\beta$ fraction of the plays on any decision path, satisfies: (i) $R(\P') \geq \beta R(\P)$ and (ii) $N(\P') \leq N(\P)$.
\end{theorem}
\begin{proof}
  Let the average reward in the initial state be $\mu$, and the prior
  is characterized by a distribution $f(\mu)$.  Consider the tree
  defined by the policy $\P$. Let $R(\P(\mu)), N(\P(\mu))$ denote the
  expected reward and the number of plays of the policy $\P$ when the
  average reward is fixed at $\mu$. We have $R(\P) = \int R(\P(\mu)) f(\mu) d\mu$.  {\bf The critical part} of the proof is the next claim which
  basically reduces to the fact that the decision to play does not
  affect the outcome.

  \begin{claim}
  \label{acctpath}
  If we execute the policy $P(\mu)$ from any node $u \in \S$ the expected reward
  $R(P(\mu))= \sum_v \mu \ y_{\mu}(v) length(v)$ where
  $\{y_{\mu}(v)\}$ is the distribution induced on the vertices of $\S$
  where $P(\mu)$ stops executing. The $length(v)$ refers to the length
  of the path (in number of edges) from $u$ to $v$.
  \end{claim}
  \begin{proof}
   We observe that if we play at
  a node $u'$, then we generate a reward $\mu$. Now the
  probability of being at $u'$ is the sum of the probability of the
  paths that pass through $u'$. 
{\small
\begin{eqnarray*}
\displaystyle R(P(\mu)) & = & \sum_{u'} \mu \ Pr[\mbox{reaching $u'$}] = \sum_{u'} \sum_{v:u' \mbox{\scriptsize is on a path to $v$, $u' \neq v$}}  \mu \ y_{\mu}(v)
 =  \sum_v \mu y_{\mu}(v)\ length(v)
\end{eqnarray*}
}
  The claim follows.
\end{proof}

\noindent {\em (Continuing Proof of Theorem~\ref{statetheorem})}
Truncating $\P(\mu)$ to $\P'(\mu)$ induces a many to one map over the
paths at which the policy stops.  Let $\{y_{\mu}(v)\}, \{y'_{\mu}(v')\}$ be
the distribution induced on the vertices of $\S$ where $\P(\mu), \P'(\mu)$
stop executing, respectively.

If the path to $v$ got truncated to $g(v)$, then $y'_{\mu}(v') = \sum_{v:
  g(v)=v'} y_{\mu}(v)$.  Further, from the statement of the Lemma,
$length(g(v)) = \beta T \geq \beta length(v)$.  Then using
Claim~\ref{acctpath}, on $\P'(\mu)$,
{\small
\begin{eqnarray*}
 R(\P'(\mu)) & = & \sum_{v'} \mu y'_{\mu}(v')\ length(v') = \sum_{v'}  \mu \ \sum_{v: g(v)=v'} y_{\mu}(v)\ length(v')
 =  \sum_{v'} \mu \sum_{v: g(v)=v'} y_{\mu}(v)\ length(g(v)) \\
& \geq & \sum_{v'} \mu \ \sum_{v: g(v)=v'} y_{\mu}(v)\ \beta\ length(v)
=  \beta \sum_{v} \mu y_{\mu}(v)\ length(v) = \beta R(\P(\mu))
\end{eqnarray*}
}
The last part follows from applying Claim~\ref{acctpath} on
$R(\P(\mu))$. Note that the rewards of each path is {\em not}
preserved up to a factor. But the Claim~\ref{acctpath} introduces an
accounting method where the reward is redistributed over the
length. Note that it was important that the outcome of the play did
not depend on the decision to play.
Integrating over $\mu$, the result of $R(\P')$ follows.  The truncation can not
  increase the expected number of plays, this proves the theorem.
\end{proof}

\subsection{Finite Horizon MAB with Instantaneous Feedback}
\label{finite}
\label{simple}
\newcommand{\finlp}{\mbox{\sc FH}}
\newcommand{\G}{\mathcal{G}}
\renewcommand {\p}{\mathbf{p}} 
As a direct application, we consider the Finite Horizon Bayesian Problem in the delay free setting. To formulate an LP relaxation, we find the optimal single-arm policy for each arm, so that the resulting ensemble of policies have expected number of plays at most $T$.  This yields: 

{\small 
\[ LP1 =  \max_{\Pp_s(i,T)} \left\{  \sum_i R(\Pp_s(i,T)) \ |\   \sum_i  N(\Pp_s(i,T))  \leq  T \right \} \]
}

The solution of LP1, denote this \finlp, is a collection of single arm policies $\P_i$ such that $\finlp \leq \sum_i R(\P_i)$ and $\sum_i N(\P_i) \le T$. We can compute these policies efficiently; see~\cite{GMS,GM09,book,Tsitsiklis}.

\medskip\noindent{\bf Rounding to a Feasible Policy:} 
The rounding scheme is now simple. (Variants have appeared before in several stochastic knapsack type optimization contexts~\cite{DGV,GM09,rudra}.) Order the arms in order of $\frac{R(\P_i)}{N(\P_i)}$. Play $\{\P_i\}$ in that order. If the decision in $\P_i$ is to quit, then we move to the next arm. If at any point we have already made $T$ plays, we stop
and quit the overall policy. 

We first prove a simple algebraic lemma.

\begin{lemma}
\label{concavechain}
Suppose $r_i/w_i$, i = 1,2,\ldots,n is a non-increasing sequence. Let $w=\sum_i w_i$. Then 
{\small
\[\sum_{i<k} r_i \left(1-\frac1w \sum_{j<i} w_i \right) \geq \frac{\left(\sum_{i<k} r_i\right)\left(\sum_{i<k} w_i\right)}{2w} + \sum_{i<k} \frac{r_i w_i}{2w} + \sum_{i<k} r_i \frac{\sum_{j\geq k} w_k}{w} \]
}
\end{lemma}
\begin{proof}
The LHS is $\frac1w$ of the area under the curve defined by the chain
of points $(0,r_1)$,$(w_1,r_1)$, $(w_1,r_1+r_2)$, $\ldots$,
$(\sum_{j<i} w_j$, $\sum_{j\leq i} r_j)$,$(\sum_{j\leq i} w_j,
\sum_{j\leq i} r_j)$,$\ldots$,$(\sum_{j<k-1} w_j, \sum_{j<k} r_j)$,
and a horizontal line segment broken into two parts $(\sum_{j<k} w_j,
\sum_{j<k} r_j)$,$(w,\sum_{j<k} r_j)$ for convenience.

The non-decreasing ordering implies that all the points in the curve 
(except the last one) are above the line $y (\sum_{j<k} w_j) = x (\sum_{j<k} r_j)$.

Therefore the area defined in the LHS dominates the sum of the areas
defined by (i) the triangle defined by $(0,0),(\sum_{j<k}
w_j,\sum_{j<k} r_j),(0,\sum_{j<k} r_j)$, (ii) the $k-1$ small
triangles defined by $(\sum_{j<i} w_j$, $\sum_{j<i} r_j)$,
$(\sum_{j<i} w_j$, $\sum_{j\leq i} r_j)$, and $(\sum_{j\leq i} w_j$, $\sum_{j\leq i} r_j)$. and (iii) The rectangle of length $\sum_{j \geq k} w_j$ and height $\sum_{j <k} r_j$.

These are the three terms on the RHS and thus the lemma follows.
\end{proof}

\begin{theorem}
\label{unitmab}
In time poly$(T,\sum_i|\S_i|)$, we can compute a $2$-approximation to the finite horizon Bayesian MAB  problem with instantaneous feedback (with budgets and arbitrary priors).
\end{theorem}
\begin{proof}
We show that the policy described above has an expected reward of $\frac12\finlp$.
Let the number of plays of arm $i$ be $T_i$. We know $\E[T_i] =
N(\P_i)$ and $\sum_i N(\P_i) = T$. We start playing arm $i$ after
$\sum_{j < i} T_j$ plays (if the sum is less than $T$); we apply the
Truncation Theorem~\ref{statetheorem} and the expected reward of
$\P_i$ continuing from this point onward is $ \left (1 - \frac{1}{T}
\min \{ T,\sum_{j<i} T_j \} \right)R(\P_i)$. Note that this is a
consequence of the independence of arm $i$ from $\sum_{j<i} T_j$.
Thus the total expected reward is
{\small
\begin{eqnarray*} 
R  & = & \E \left[ \sum_i \left(1 - \frac{1}{T} \min \{ T, \sum_{j<i} T_j\} \right)R(\P_i) \right] \geq  
\E \left[ \sum_i \left(1 - \frac{1}{T}\sum_{j<i} T_j \right)R(\P_i) \right] \\
& = & 
\sum_i \left (1 - \frac{1}{T}\sum_{j<i} \E[T_j] \right)R(\P_i) = \sum_i R(\P_i)\left( 1 -\frac{\sum_{j<i} N(\P_j)}{T} \right) 
\end{eqnarray*}
}
By Lemma~\ref{concavechain} (with $r_i =R(\P_i)$, 
$w_i=N(\P_i)$ and $k=n+1$) this is at least $\frac12 \sum_j
R(\P_j)=\frac12\finlp$.
\end{proof}

\paragraph{Tight example of the analysis.}
We show that the gap of the optimum policy and LP1 is a
factor of $2 - O(\frac1n)$, even for unit length plays. Consider the
following situation: We have two ``types'' of arms. The type I arm
gives a reward $0$ with probability $a = 1/n$ and $1$ otherwise. The type
II arm always gives a reward $0$.  We have $n$ independent arms. Each
has an identical prior distribution of being type I with probability
$p = 1/n^2$ and type II otherwise. Set $T = n$.

Consider the symmetric LP solution that allocates one play to each arm; if it observed a $1$, it plays the arm for $n$ steps. The expected number of plays made is $n + O(1/n)$, and the expected reward is $n \times 1/n = 1$. Therefore,  $LP1 \ge 1 - O(1/n)$.

Consider the optimum policy. We first observe that if the policy ever
sees a reward $1$ then the optimum policy has found one of the type II
arms, and the policy will continue to play this arm for the rest of
the time horizon. At any point of time before the time horizon, since
$T=n$, there is always at least one arm which has not been played
yet. Suppose we play an arm and observe the reward $0$, then the
posterior probability of this arm being type II increased. So the
optimum policy should not prefer this currently played arm over an
unplayed arm. Thus the optimum policy would be to order the arms
arbitrarily and make a single play on every new arm. If the outcome is
$0$, the policy quits, otherwise the policy keeps playing the arm for
the rest of the horizon. The reward of the optimum policy can thus be
bounded by $\sum_{x=0}^{T-1} ap(1 + (T-x-1)a) = pa^2T(T+1)/2 +
(1-a)/n= \frac12 + O(\frac1n)$. Thus the gap is a factor of $2 -
O(\frac1n)$. 

\section{Multi-armed Bandits with Delayed Feedback}
\label{part2}
 Define $Q_i = T/\delta_i$. We assume that $\delta_i = o(T/\log T)$, which implies that the horizon is slightly separated from the delays.   We show the following theorem. 
 
\begin{theorem}
\label{thmone}
Assuming $\delta_i = o(T/\log T)$, there is a constant factor approximation to the Bayesian MAB problem with delayed feedback. The running time for computing this policy is $\mbox{poly}(T,\sum_i|\S_i|)$. 
\end{theorem} 

As in Section~\ref{finite}, we will use (LP1) to bound of the reward of the best collection of single-arm policies. However, it is not clear how to solve the above LP in polynomial time, since each $\Pp_s(i,T)$ has description which is exponential in the delay parameter $\delta_i$. We first simplify the structure of the single-arm policies to enable poly-time computation (Step 1). Even after this, it could happen that the policies make most of their plays after $T/2$ steps, so that the truncation theorem cannot be applied directly. We then show how to compact the policies (Step 2) to reduce the horizon - this compaction uses an interesting simulation argument. We then truncate the policies, and solve the LP relaxation over these well-structured policies (Step 3). We finally design a scheduling algorithm for combining these single-arm policies (Step 4) and analyze the approximation ratio.

\subsection{Step 1: Block Structured Policies}
\begin{definition}
A single-arm policy is said to be {\em Block Structured} if the policy executes in phases of size $(2\delta_{i}+1)$. At the start of each
phase (or block), the policy makes at most $\delta_{i}+1$ consecutive
plays. The policy then waits for the rest of the block in order to
obtain feedback on these plays, and then moves to the next block. A
block is defined to be {\em full} if exactly $\delta_i+1$ plays are
made in it.
\end{definition}

We first show  that all single-arm policies can be replaced with block structured policies while violating the time horizon by a
constant factor. The idea behind this proof is simple -- we simply insert delays of length $\delta_i$ after every chunk of plays of length $\delta_i$.  

\begin{lemma}
\label{lemone}
Any policy $\Pp(i,T)$ can be converted it to a Block Structured policy  $\Pp'(i,2T)$ such that $R(\Pp(i,T)) \leq R(\Pp'(i,2T))$  and $N(\Pp'(i,2T)) \leq N(\Pp(i,T)) $
\end{lemma}
\begin{proof}
  We can assume that the policy makes a play at the very first time
  step, because we can eliminate any wait without any change of
  behavior of the policy.  

Consider the actions of $\Pp(i,T)$ for the
  first $\delta_{i}+1$ steps, the result of any play in
  these steps is not known before all these plays are
  made. An equivalent policy $\Pp'$ simulates $\Pp$ for the first
  $\delta_{i}+1$ steps, and then waits for $\delta_{i}$ steps, for a
  total of $2\delta_{i}+1$ steps.  This ensures that $\Pp'$ knows the
  outcome of the plays before the next block begins.

Now consider the steps from $\delta_{i}+2$ to $2\delta_{i}+3$ of
  $\Pp(i,T)$. As $\Pp$ is executed, it makes some plays
  possibly in an adaptive fashion based on the outcome of the plays in
  the previous $\delta_{i} + 1$ steps, but not on the current
  $\delta_{i}+1$ steps.  $\Pp'$ however knows the outcome of the
  previous plays, and can simulate $\Pp$ for these $\delta_{i}+1$
  steps and then wait for $\delta_{i}$ steps again.  It is immediate
  to observe that $\Pp'$ can simulate $\Pp$ at the cost of increasing
  the horizon by a factor of $\frac{2\delta_{i}+1}{\delta_{i}+1} < 2$.
  Observe that in each block of $2\delta_{i}+1$, $\Pp'$ can
  make all the plays consecutively at the start of the block without
  any change in behavior. The budgets are also respected in this process. 
  This proves the lemma.
\end{proof}

\subsection{Step 2: Well-Structured Policies and Simulation}
\label{appc}
The block-structured policies constructed above still suffer from the drawback that too many plays can be made close to the horizon $T$. Initially, the policy can be conservative and play very few times in each block waiting for feedback. In this part, we show how to compact such policies while preserving the reward and making sure the expected number of plays only increases by a constant factor. Our technique uses the idea of {\em simulation} - we make more plays initially, but hold on to the outcomes of the extra plays. When the original policy makes plays in a subsequent block, we eliminate these plays and instead use the outcomes of the saved up plays. As with the Truncation theorem, this argument crucially uses the martingale property that the reward of the arm are $i.i.d.$ draws from the same unknown underlying distribution regardless of when the plays are made. 

\begin{definition}
Define a block-structured policy to be {\em $c$-delay-free} for $c \le 1$ if the first time the policy encounters a block with at least $c \delta_i $ plays, it plays every step (without waiting) beyond this point (using feedback from $\delta_i$ plays ago) until it stops executing.
\end{definition}

Intuitively, we are eliminating delays if a policy makes a sufficiently large number of plays in a block. This step not only compacts the policy, but also shrinks the state space significantly, since if the policy executes in a delay-free fashion, we can pretend the feedback is instantaneous. (In reality, the feedback is from $\delta_i$ plays ago, but we can easily couple the two executions.) We now show that any block-structured policy can be made $c$-delay-free.

\begin{lemma}
\label{lemtwo} 
Given any Block Structured policy $\Pp(i,2T)$ we can  construct a $c$-delay-free Block-structured policy $\Pp'(i,2T)$, such that  $R(\Pp(i,2T)) \leq R(\Pp'(i,2T))$ and $N(\Pp'(i,2T)) \leq (1+ \frac{1}{c}) N(\Pp(i,2T))$.
\end{lemma}
\begin{proof}
Consider the first time the policy encounters a block with $c \delta_i$ plays on some decision path. We play the arm continuously beyond this point and {\em simulate} the behavior of the original policy $\Pp$. For each contiguous play the policy $\Pp'$ makes, this outcome is available after $\delta_i$ steps. Simulate the next play of $\Pp$ using this outcome, until $\Pp$ stops execution. Clearly,  $\Pp'$ makes at most $\delta_i$ additional plays than $\Pp$ on any decision path; since $\Pp$ made at least $c \delta_i$ plays, this shows that  $N(\Pp'(i,2T)) \leq (1+ \frac{1}{c}) N(\Pp(i,2T))$. Since the execution of $\Pp'$ is coupled play-by-play with the execution of $\Pp$, it is clear that $R(\Pp(i,2T)) \leq R(\Pp'(i,2T))$.
\end{proof}

The policies constructed above still suffer from being too sparse if the number of plays in each block is less than $c \delta_i$. In this step, we will compact these policies further to retain mostly the full blocks. Note that any policy uses at most $Q_i = \frac{T}{\delta_i} = \omega(\log \delta_i)$ blocks. 

\begin{definition}
For constant $\alpha < 1$, define a $c$-delay-free block-structured single-arm policy $\P$ to be $(\alpha,c)$-{\em  well-structured} if after encountering at most $q = (\alpha + o(1)) Q_i$ blocks, the  policy switches to playing continuously, {\em i.e.}, executing in delay-free mode.
\end{definition}

The next lemma further compacts the block-structured policies into well-structured policies. 

\begin{lemma}
\label{lemthree}
For any $\alpha < 1$ and $c \le \frac{\alpha}{\alpha + 2}$, given a $c$-delay-free policy $\Pp(i,2T)$, there is a $(\alpha,c)$-well-structured policy $\Pp'$ such that $R(\Pp(i,2T)) \leq  R(\Pp'(i,2T))$ and $N(\Pp'(i,2T)) \leq (1+\frac{2}{\alpha})  N(\Pp(i,2T))$.
\end{lemma}
\begin{proof}
Consider the execution of $\Pp$. All blocks in this policy have at most $c \delta_i$ plays, unless the policy is executing in delay-free mode.  We consider the execution of the original policy, and show a coupled execution in the new policy $\Pp'$ so that if on a decision path, $\Pp$ used $k$ blocks, then the number of blocks on the same decision path in $\Pp'$ is $\alpha k + O(\log \delta_i)$. We group blocks into size classes; size class $s$ has blocks whose number of plays lies in $[2^s,2^{s+1}]$.

We couple the executions as follows: Consider the decision tree $\Pp$ top-down. Increase the plays in the root block $r$ by a factor of $(1+\frac{2}{\alpha})$; suppose there were $x$ plays originally in this block, and the size class of this block is $s$. Consider the first $\frac{1}{\alpha}$ blocks  on each decision path downstream of $r$ whose size class is at most $s$; eliminate these blocks and use the outcomes of the $\frac{2x}{\alpha}$ extra plays made at $r$ to simulate the behavior of $\Pp$ in these blocks as follows. Simply store the outcomes of the extra plays at $r$ (without updating the prior), and use these outcomes when $\Pp$ makes those plays. Since the underlying reward distribution is fixed, these plays will be stochastically identical. We note that $\Pp'$ does not really need to ``store" the outcomes - it can simply update the prior at $r$ and follow the execution of $\Pp$ as if it had not updated it. Again, since the reward process is a martingale, both these executions are stochastically identical. 

On each decision path, either we used $r$ to eliminate $\frac{1}{\alpha}$ blocks, or there are no blocks left with size class at most $s$. Now repeat this procedure on the roots of each decision sub-tree downstream of $r$. When this procedure terminates, consider any decision path, and mark the leftover blocks whose increase in plays could not eliminate exactly $\frac{1}{\alpha}$ blocks.  For any size class $s$, it is clear that there can be at one such block, else the earlier block could have eliminated the later one. Therefore, if a decision path at $k$ blocks in $\Pp$, the above procedure shrinks the number of blocks to at most $\alpha k + \frac{\log_2 \delta_i}{\alpha}$, since there are at most $\log_2 \delta$ size classes. This procedure increases the expected number of plays by at most a factor of $1+\frac{2}{\alpha}$, and we have a guarantee that each decision path makes at most $ (\alpha + o(1)) Q_i$ plays, since $\log \delta_i = o(Q_i)$.
\end{proof}

\renewcommand{\SS}{{Super}} 
\newcommand{\tp}{{\tilde{p}}} 

\subsection{Step 3: Truncation and Solving the LP Relaxation}
We now truncate the horizon of the well-structured policies losing a constant factor in the reward.

\begin{lemma}
\label{lemfour}
Given any $(1/8,1/17)$-well-structured single arm policy $\P = \Pp(i,2T)$, there is an identical policy that stops execution after $\frac{T}{2}$ times steps (denoted by $\P'$), that satisfies the following: (i) $R(\P') =  \frac{R(\P)}{8}$ and (ii) $N(\P') \leq N(\P)$.
\end{lemma}
\begin{proof}
On any decision path in $\P$, there can be at most $Q_i(1/8 + o(1))$ blocks before the policy executes in delay-free mode. These blocks take up at most $T/4$ time steps, and the plays beyond that are delay-free. This implies that if the policy's horizon is truncated to $T/2$ time steps, and the policy has not stopped executing, it makes at least $T/4$ contiguous plays, whereas $\P$ could have made at most $2T$ plays. Therefore, on any decision path, the number of plays is reduced by at most a factor of $8$, and by the Truncation theorem~\ref{statetheorem}, the lemma follows.
\end{proof}

\noindent We term well-structured policies that terminate on encountering more than $T/2$ time steps as {\bf truncated well-structured} policies. The above lemmas can be summarized as:

\begin{theorem}
\label{trunc}
$(LP1)$ has a $1/\alpha$-approximation  over truncated $(\alpha,\frac{\alpha}{\alpha+2})$-well structured policies for $\alpha \le 1/8$.  This LP has the relaxed constraint:  $\sum_i N(\Pp_s(i,T/2)) \leq \gamma T$ where $\gamma = 2(1+1/\alpha)(1+2/\alpha)$. 
\end{theorem}

We now formulate a new LP relaxation over truncated well-structured policies constructed in Theorem~\ref{trunc}. This new LP is polynomial over $|\S_i|$ and we avoid the exponential dependence on $\delta_i$ due to the above transformations. Towards this end, we explicitly write out the state space of a truncated $(\alpha,c)$-well-structured policy for a single arm. For $u \in \S_i$, we create a state $W(u,t)$ for $1 \le t \le T/2$. The meaning of this state is as follows: The posterior at the beginning of a block is $u \in \S_i$, and $t$ time steps have elapsed so far.  This block is either regular (in which case, at most $\alpha Q_i$ blocks have elapsed so far), or no-delay. For no-delay blocks, we can assume the block-length is $1$ and feedback is instantaneously available; this unifies the presentation of the LP. (The feedback is actually from plays made $\delta_i$ steps ago, but we can couple the two executions.)

For any state $\sigma = W(u,t)$ of arm $i$, define the following quantities. These quantities are an easy computation from the description of $\S_i$, and the details are omitted.
\begin{enumerate}
\item Let $r_i(\sigma,\ell)$ denote the expected reward obtained when   consecutive $\ell$ plays are made at state $\sigma$, and feedback is  obtained at the very end. Note that for regular state $\sigma$, we have $0 \le \ell c \delta_i$. For a no-delay state, $0   \le \ell \le 1$.
\item Let $p_i(\sigma,\sigma',\ell)$ denote the probability that if  the state at the beginning of a block is $\sigma$, and $\ell$ plays are  made in the block, the state at the beginning of the next block is  $\sigma'$. 
\end{enumerate}

We will formulate an LP to find a randomized truncated well-structured policy $\Pp$. Define the following variables over the decision tree of the policy:

\begin{itemize}
\item $x_{i\sigma}$: the probability that the state for arm $i$ at the start of a block is $\sigma$.
\item $y_{i\sigma \ell}$: probability that the policy  for arm $i$ makes $\ell$ consecutive plays starting at state  $\sigma$.
\end{itemize}

We have the following LP relaxation, which simply encodes finding one randomized well-structured policy per arm so that the expected number of plays made is at most $\gamma T$. The objective to this LP is $\Omega(OPT)$ as a consequence of Theorem~\ref{trunc}. 

{\small
\[ LP2 = \mbox{Maximize}  \sum_{i,\sigma,\ell} r_i(\sigma,\ell) y_{i\sigma \ell} \]
\[  \begin{array}{rcll}
\sum_{\sigma} x_{i\sigma} & \le & 1 & \forall i \\
\sum_{\ell} y_{i \sigma \ell} & \le & x_{i \sigma} & \forall i, \sigma \\
\sum_{\sigma,\ell} y_{i \sigma \ell} p_i(\sigma, \sigma', \ell) & = & x_{i \sigma'} & \forall i, \sigma' \\
\sum_{i,\sigma,\ell} \ell y_{i \sigma \ell} & \le & \gamma T & \\
x_{i \sigma}, y_{i \sigma \ell} & \ge & 0  & \forall i, \sigma, \ell
\end{array} \]
}

\newcommand{\Ppr}{{{\cal P}^{r}}}

\medskip
\noindent{\bf The LP policy:} 
We scale down the LP solution by a factor of $\gamma$ so that the final constraint becomes $\sum_{i,\sigma,\ell} \ell y_{i \sigma   \ell} \le T$; denote the new LP as $(LP2s)$ for convenience. $(LP2s)$ yields one randomized well-structured policy $\Ppr(i,T/2)$ for each arm $i$. This policy is succinctly described in Figure~\ref{fig1a}.

\begin{figure*}[htbp]
\centerline{\framebox{\small
\begin{minipage}{6in}
{\bf Policy $\Ppr(i,T/2)$:}
\begin{itemize}
\item If the state at the beginning of a block is $\sigma$: 
\begin{enumerate}
\item Choose $n$ with probability $\frac{y_{i \sigma n}}{x_{i \sigma}}$, and make $n$ plays in the current block.
\item Wait till the end of the block; obtain feedback for the $n$ plays; and update state.
\end{enumerate}
\end{itemize}
\end{minipage}}
}
\caption{\label{fig1a} Single-arm policy $\Ppr(i,T/2)$.}
\end{figure*}

\subsection{Step 4: Priority Based Combining of Different Policies}
At this point we have a collection of randomized policies $\Ppr(i,T/2)$ such that $\sum_i N(\Ppr(i,T/2)) \leq T$ and $\sum_{i} R(\Ppr(i,T/2)) = \Omega(OPT)$. We now show how to  combine the single arm policies $\Ppr(i,T/2)$ to achieve a globally approximate and feasible solution.  This step is not complicated but we note that multiple policies must remain active in the combination. This aspect necessitates a novel priority based scheme for combining the policies. 

Consider the execution of $\Ppr(i,T/2)$. We describe the arm as {\em active} if it is either making plays or is at the beginning of a block where it can make plays; and {\em passive} if it is waiting for feedback on the plays within the block. Any arm which completed its waiting for the feedback turns from passive to active mode. The final policy is shown in Figure~\ref{fig2a}.

\begin{figure*}[htbp]
\centerline{\framebox{\small
\begin{minipage}{6in}
\begin{enumerate}
\item Choose an arbitrary order the arms $\{i\}$ 
  denoted by $\pi$. 
\item Each arm ``participates'' with  probability $1/4$. Initially, all participating arms are active.
\item On each new play, among all participating arms:
\begin{enumerate}
\item Find the lowest  rank arm, say $i'$, that is active.  
\item Allocate the current play to $i'$ according to the policy $\Ppr(i',T/2)$. 
\item (As a result of this allocation, the arm may become passive and wait for feedback.) 
\end{enumerate} 
\end{enumerate}
\end{minipage}}
}
\caption{\label{fig2a} The final policy {\bf Combine}.}
\end{figure*}

The next lemma follows from an application of Markov's inequality in the same spirit as prior work without delayed feedback~\cite{GM09}. 

\begin{lemma}
\label{acct}
The expected contribution of $i$ from the combined policy is at least $R(\Ppr(i,T/2))/8$.
\end{lemma}
\begin{proof} 
First observe that $i$ being active does not impact any play of $G_i$.  The expected number of plays made by $G_i$  is $\sum_{i' \in G_i} N(\Ppr(i',T/2))/4 \leq T/4$. Thus by Markov's Inequality, the arms in $G_i$ make $T/2$ plays  with probability at most $1/2$. Next note that with probability $1/8$, arm $i$ is active  and in a scenario where the number of plays made by bidders before $i$  is at most $T/2$. This means {\bf every} decision path of  $\Ppr(i,T/2)$ can complete in this scenario (since $i$ gets higher  priority that all bidders following it). Thus with probability $1/8$ the  policy $\Ppr(i,T/2))$ executes without any interference.  Thus the   expected contribution from arm $i$ is at least   $R(\Ppr(i,T/2))/8$. 
\end{proof}

The above lemma shows Theorem~\ref{thmone} by linearity of expectation, concluding the  analysis.

\appendix

\newpage

\section{Budgeted Allocation with Delayed Feedback}
\label{part3}
Recall that the budgeted allocation problem is a collection of MAB
instances, one for each impression type. For each impression of type
$j$, the MAB instance $M(T_j)$ is over $T_j$ arrivals (or plays) of
this impression. (Note that $\sum_j T_j = T$.) The arms for this
impression are denoted $(i,j)$, where $i$ denotes the bidders. If a
play along $(i,j)$ is successful (which is observed after
$\delta_{ij}$ arrivals of impressions of type $j$), bidder $i$ accrues
reward $b_{ij}$; the probability of success is drawn from an
independent prior $\D_{ij}$. The arms are connected by the budget
constraint stating that the total reward accrued by bidder $i$ across
all MAB instances is at most $B_i$.  

\medskip
\noindent {\bf Algorithm and Analysis:} The overall algorithm is
similar to that for a single MAB instance.
\begin{enumerate}
\item  The initial LP relaxation (that we don't solve) encodes
constructing one randomized policy per arm $(i,j)$ such that: (i) The
total expected number of plays for MAB $M_j$ is at most $T_j$; and
(ii) The total expected reward accrued for bidder $i$ is at most
$B_i$. Note that the latter constraint was not present in $(LP1)$ in
the previous section. 
\item With a loss of factor $2$ in the reward, we define ``shadow budgets'' $B_{ij} \leq B_i$ 
such that $B_{ij}/b_{ij}$ are integers, as discussed in Section~\ref{part1}.
\item For each $(i,j)$, we transform the single-arm policy (with the
  budgets $B_{ij}$) just as in the previous section into a
  well-structured policy using $T_j/2$ impressions.
\item We solve an LP over truncated well-structured policies, with the
additional constraint that the total reward accrued for bidder $i$ is
at most $B_i$. Let the LP objective value be $M$.
\item To construct the global policy, we run the policy {\bf Combine} independently
for each $j$. Let the randomized well-structured policy for arm
$(i,j)$ that is constructed by the LP be denoted $\Ppr(i,j,T_j/2)$,
and let $Z_{ij}$ denote the reward achieved. Note $Z_{ij} \leq B_{ij} \leq B_i$.
\end{enumerate}

\noindent The analysis in Section~\ref{part1} shows that the policy {\bf
  Combine} yields expected reward $\sum_{i,j} \E[Z_{ij}] = \Omega(M)$.
However we can only extract a reward of $\min\{B_i,\sum_j
Z_{ij}\}$ from advertiser $i$. We show: 

\begin{lemma}\label{mincount}
For a collection of random variables $\{Z_{ij}\}$, if $\sum_j
\E[Z_{ij}] \leq B_i$ and $Z_{ij} \leq B_i$ then $\E[ \min \{ B_i, \sum_j Z_{ij} \}] \geq
\frac12 \sum_j \E[Z_{ij}]$.
\end{lemma}
\begin{proof}
Consider a process where advertiser $i$ settles the bill for $j$, before proceeding to the bill for $j+1$. If the payout on $j$ is $Y_j$ then $Y_j = \min \left \{ Z_{ij}, \max \{B_i - \sum_{j'<j} Z_{ij'}, 0 \} \right \}$. 
{\small
\[ Y_j 
= B_i \min \left \{ \frac{Z_{ij}}{ B_i},  \max \left( 1 - \frac{\sum_{j'<j} Z_{ij'}}{B_i}, 0 \right) \right \} 
\geq \max \left(1 - \frac{\sum_{j'<j} Z_{ij'}}{B_i},0 \right)Z_{ij} \]
}
\noindent The last part follows from $\min\{a,b\} \geq ab$,  
for any $0\leq a,b \leq 1$. 
Clearly $\sum_j Y_j \leq B_i$ in every scenario. 
And as a consequence of $Y_{j}$ being independent of $Z_{ij'}$ for $j'<j$: 
{\small
\[ 
\E[Y_j] \geq \E \left[ \max \left(1 - \frac{\sum_{j'<j} Z_{ij'}}{B_i},0 \right) \right ] \E[Z_{ij}] \geq \left(1 - \frac{\sum_{j'<j} \E[Z_{ij'}]}{B_i} \right) \E[Z_{ij}] 
\]
}
\noindent We can estimate $\sum_j \E[Y_j] $ to be at least $\frac1{B_i}$ times the area
under the staircase defined by 
{\small
\[ (0,\E[Z_{i1}]),\ldots,
(\sum_{j'<j} \E[Z_{ij}], \sum_{j'\leq j} \E[Z_{i1}]),(\sum_{j'\leq j} \E[Z_{ij'}], \sum_{j'\leq j} \E[Z_{ij'}]),\ldots,(\sum_j  \E[Z_{ij}], \sum_{j} \E[Z_{ij}]),(B_i,\sum_{j} \E[Z_{ij}]) \] 
}

This area includes the triangle defined $(0,0)$, $(B_i,0)$ and
$(B_i,\sum_{j} \E[Z_{ij}])$ and thus the reward is therefore at least 
$\frac1{B_i} \frac12 B_i \sum_j \E[Z_{ij}] = \frac12 \sum_j \E[Z_{ij}]$. 
This proves the lemma.
\end{proof}

As a consequence of Lemma~\ref{mincount}, and that $M=\Omega(OPT)$ we have:
\begin{theorem}
\label{thmtwo}
  If $T_j \geq 48(8\delta_{ij}+2)\log T_j$ for all $,j$ then in time
  polynomial in $n,m$, and $\max_{i,j} |\S_{ij}|$, we can compute a policy whose
  expected reward is a $O(1)$ approximation to the expected reward of
  the optimal policy for the Bayesian budgeted allocation problem with delayed
  feedback.
\end{theorem}

\section{Example of Priors: Beta Distributions} 
\label{prelim2}
Consider a fixed but unknown distribution
over two outcomes (success or failure of treatment, click or no click,
conversion or no conversion). This is a Bernoulli$(1,\theta)$ trial
where $\theta$ is unknown.  One way of encoding the uncertainty about
$\theta$ is to use a distribution which encodes our ``prior'' belief
regarding $\theta$. However on getting a new sample, we update our
belief to a ``posterior distribution'' using Bayes' rule. Note that we
need a family of distributions and associated update rules to specify
the prior.  For two valued outcomes, one such family is the Beta
distribution, referred to as the conjugate prior of the Bernoulli
distribution.  A Beta distribution with parameters $\alpha_0,\alpha_1
\in \mathbb{Z}^+$, which we denote $Beta(\alpha_1,\alpha_0)$ has
p.d.f. of the form $c\theta^{\alpha_1-1} (1-\theta)^{\alpha_0-1}$,
where $c$ is a normalizing constant. $Beta(1,1)$ is the uniform
distribution.  The distribution $Beta(\alpha_1,\alpha_0)$ corresponds
to the current (posterior) distribution over the possible values of
$\theta$ after having observed $(\alpha_1-1)$ $1$'s and $(\alpha_0-1)$
$0$'s, starting from the belief that $\theta$ was uniform, distributed
as $Beta(1,1)$.  Given the distribution $Beta(\alpha_1,\alpha_0)$ as
our prior, the expected value of $\theta$ is
$\frac{\alpha_1}{\alpha_1+\alpha_0}$.  In this example, updating the
prior on a sample is straightforward. On seeing a $1$, the posterior
(of the current sample, and the prior for the next sample) is
$Beta(\alpha_1+1,\alpha_0)$. On seeing a $0$, the new distribution is
$ Beta(\alpha_1,\alpha_0+1)$.  As stated earlier, the hard case for
our algorithms (and the case that is typical in practice) is when the
input to the problem is a set of arms $\{i\}$ with priors $\D_i \sim
Beta(\alpha_{1i},\alpha_{0i})$ where $ \alpha_{0i} \gg \alpha_{1i}$
which corresponds to a set of {\bf poor prior expectations} of the
arms.

\paragraph{Representing Priors as a DAG.}  Given a family of priors we
model them as a directed acyclic graph (DAG) where the the vertices
correspond to the priors and posteriors.  The root $\rho$ encodes the
initial prior $\D$. The children of $\rho$ correspond to the possible
observations of a sample drawn from $\D$. In the case of Beta priors
and two valued outcomes, a node $u$ corresponding to
$Beta(\alpha_1,\alpha_0)$ would have two children corresponding to
$Beta(\alpha_1+1,\alpha_0),Beta(\alpha_1,\alpha_0+1)$ and the edges
will be labeled $1,0$ respectively. The probability of taking the edge
labeled $1$ from $u$ to the node $v$ corresponding to
$Beta(\alpha_1+1,\alpha_0)$ would be the probability of observing a
$1$ given the prior $Beta(\alpha_1,\alpha_0)$, which is
$\frac{\alpha_1}{\alpha_1+\alpha_0}$. We will define
$\p_{uv}=\frac{\alpha_1}{\alpha_1+\alpha_0}$. At each node $u$ the
probability of observing $1$ corresponds to the ``reward'' $\mu_u$.
Thus the DAG represents the evolution of priors/posteriors completely,
and further (for two valued outcomes) each node is represented by a
pair of numbers. Therefore for a horizon of size $T$, the entire
information is captured by a DAG of
size $O(T^2)$.

\medskip
  \noindent{\bf Martingale Property:} A {\em critical} observation, true
  of all conjugate prior distributions and the resulting DAGs, is that
  we have a {\bf Martingale property:} $\mu_u = \sum_{v \in \S_i}
  \mu_v \p_{uv}$ (summing over the children of $u$). {\em Therefore
    this gives us a way of switching between two different views of
    the same evolution of the posterior space}: the first view is where
  the $p_i$ (equivalently $D_i$) is chosen according to $\D_i$ and we
  observe a collection of trajectories and take expectation over them,
  and the second view is where we have a Martingale process of
  the posterior spaces in the DAG.

\end{document}